\documentclass[11pt]{article}
\newcommand{\LIPICS}[1]{}
\newcommand{\PAPER}[1]{#1}

\PAPER{

\usepackage{fullpage,amsmath,amsthm,amssymb,hyperref,graphicx}

\newtheorem{theorem}{Theorem}[section]
\newtheorem{lemma}[theorem]{Lemma}
\newtheorem{corollary}[theorem]{Corollary}
\newtheorem{problem}{Problem}[section]

\title{Faster Algorithms for Largest Empty Rectangles and Boxes}

\author{Timothy M. Chan\thanks{
  Department of Computer Science,
  University of Illinois at Urbana-Champaign
  (tmc@illinois.edu).  
  This research has been supported in part by NSF Grant CCF-1814026.
}}

}

\newcommand{\OO}{\widetilde{O}}
\newcommand{\R}{\mathbb{R}}
\newcommand{\eps}{\varepsilon}
\newcommand{\D}{\Delta}
\newcommand{\up}[1]{\left\lceil #1\right\rceil}
\newcommand{\down}[1]{\left\lfloor #1\right\rfloor}
\newcommand{\IGNORE}[1]{}

\begin{document}

\maketitle

\begin{abstract}
We revisit a classical problem in computational geometry: finding the largest-volume axis-aligned empty box (inside a given bounding box) amidst $n$ given points in $d$ dimensions.  Previously, the best algorithms known have running time $O(n\log^2n)$ for $d=2$ (by Aggarwal and Suri [SoCG'87]) and near $n^d$ for $d\ge 3$.  We describe faster algorithms with running time

\begin{itemize}
\item $O(n2^{O(\log^*n)}\log n)$ for $d=2$,
\item $O(n^{2.5+o(1)})$ time for $d=3$, and
\item $\widetilde{O}(n^{(5d+2)/6})$ time for any constant $d\ge 4$.
\end{itemize}

\noindent
To obtain the higher-dimensional result, we adapt and extend previous techniques for Klee's measure problem to optimize certain objective functions over the complement of a union of orthants.
% by adapting ideas 
%previously used for Klee's measure problem.
%We also obtain results for related problems, including finding
%the largest-volume empty box with a fixed corner, and 
%the largest- or smallest-volume box enclosing $k$ points.
\end{abstract}

\section{Introduction}

\paragraph{Two dimensions.}
In the first part of this paper, we tackle the \emph{largest empty rectangle} problem:
Given a set $P$ of $n$ points in the plane and a fixed rectangle $B_0$, find the largest rectangle $B\subset B_0$ such that $B$ does not contain
any points of $P$ in its interior.  Here and throughout this paper,
a ``rectangle'' refers to an axis-parallel rectangle; and unless
stated otherwise,
``largest'' refers to maximizing the area.

The problem has been studied since the early years of computational geometry.  While 
similar basic problems such as largest empty circle or largest empty square can be solved efficiently using Voronoi diagrams, the largest empty rectangle
problem seems more challenging.  The earliest reference on the 2D problem 
appears to be
by Naamad, Lee, and Hsu in 1984~\cite{Naa}, who gave a quadratic-time algorithm.  In 1986, Chazelle, Drysdale,
and Lee~\cite{ChDrLe} obtained an $O(n\log^3 n)$-time algorithm.
Subsequently, at SoCG'87, Aggarwal and Suri~\cite{AggSur} 
presented another algorithm requiring
$O(n\log^3n)$ time, followed by a more complicated second algorithm
requiring $O(n\log^2n)$ time.
The $O(n\log^2n)$ worst-case bound has not been improved since.\footnote{Aggarwal and Suri's
first algorithm can be sped up to run in near $O(n\log^2n)$ time as well, 
since it relied on a subroutine for finding row minima in Monge staircase matrices, a problem for which  improved results were later found~\cite{KlaKle,ChaSODA21}; but these results do not
appear to lower the cost of Aggarwal and Suri's second algorithm.}

%Aggarwal and Suri: first algorithm $O(n\log^3n)$ can be improved to $O(n\log^2n)$ using matrix searching for staircase matrices;
%second algorithm $O(n\log^2n)$ is rather complicated, hard to follow...

A few results on related questions have been given.
Dumitrescu and Jiang~\cite{DumJia} examined
the combinatorial problem of determining the worst-case number of maximum-area empty rectangles 
and proved an $O(n2^{\alpha(n)}\log n)$ upper bound;
their proof does not appear to have any implication to the 
algorithmic problem of finding a maximum-area empty rectangle.
If the objective is changed to maximizing the perimeter, the problem
is a little easier and an optimal $O(n\log n)$-time algorithm can already be
found in Aggarwal and Suri's paper~\cite{AggSur}.  Another related
problem of computing a maximum-area rectangle contained in a polygon
has also been explored~\cite{DaMiRo}.

We obtain a new randomized algorithm that finds the maximum-area empty rectangle in
$O(n2^{O(\log^*n)} \log n)$ expected time.  This is not only
an improvement of almost a full logarithmic factor over the
previous 33-year-old bound, but is also close to optimal,
except for the slow-growing iterated-logarithmic-like factor
(as $\Omega(n\log n)$ is a lower bound in the algebraic decision tree
model).  

Our solution interestingly uses \emph{interval trees} to efficiently
divide
the problem into subproblems of logarithmic size, yielding a recursion
with $O(\log^*n)$ depth.

\paragraph{Higher dimensions.}
The higher-dimensional analog of the problem is \emph{largest empty box}:
Given a set $P$ of $n$ points in $\R^d$ and a fixed box $B_0$, find the largest box $B\subset B_0$ such that $B$ does not contain
any points of $P$ in its interior.  Here and throughout this paper,
a ``box'' refers to an axis-parallel hyperrectangle; and unless
stated otherwise,
``largest'' refers to maximizing the volume.

Several papers~\cite{DumJia13,DumJia,DumJia16,Ull} have
studied related questions in
higher dimensions, e.g., 
proving combinatorial bounds on the number of optimal boxes, or
proving extremal bounds on the volume, or
designing approximation algorithms.
For the original (exact) computational problem, it is not difficult
to obtain an algorithm that finds the largest empty box in $\OO(n^d)$ time (for example,
as was done by Backers and Keil~\cite{BacKei}).\footnote{Throughout the paper,
$\OO$ notation hides polylogarithmic factor.}
At the end of their SoCG'16 paper,
Dumitrescu and Jiang~\cite{DumJia} explicitly asked whether
a faster algorithm is possible:
\begin{quote}
``Can a maximum empty box in $\R^d$
for some fixed $d\ge 3$ be computed in 
$O(n^{\gamma_d})$ time
for some constant $\gamma_d < d$?''
\end{quote}
Dumitrescu and Jiang attempted to give a subcubic algorithm 
for the 3D problem, but their conditional solution required a sublinear-time
dynamic data structure for finding the 2D maximum empty rectangles containing a query point---currently, the existence of such a data structure is not known.

On the lower bound side, Giannopoulos, Knauer,
Wahlstr\"om, and Werner~\cite{Gia} proved that the largest
empty box problem is $W[1]$-hard with respect to the dimension.
This implies a conditional lower bound of $\Omega(n^{\beta d})$ for some
absolute constant $\beta>0$, 
assuming a popular conjecture on the hardness of the clique problem.

We answer the above question affirmatively.  For $d=3$, we give
an $O(n^{5/2+\eps})$-time algorithm, where $\eps>0$ is
an arbitrarily
small constant.  For higher constant $d\ge 4$, we obtain
an algorithm with an intriguing time bound that improves over $n^d$ even more dramatically: $\OO(n^{(5d+2)/6})$.  
For example, the bound is $O(n^{3.667})$ for $d=4$, $\OO(n^{4.5})$ for $d=5$, and $O(n^{8.667})$ for $d=10$.

Not too surprisingly, our 3D algorithm achieves subcubic complexity by
applying standard range searching data structures (though the application is not be immediately obvious).
Dynamic data structures are not used. 
 
The techniques for our higher-dimensional algorithm
are perhaps more original and significant, with potential impact to other
problems.  We first transform the largest empty box problem 
into a problem about a union of $n$ orthants in $D=2d$ dimensions (the
transformation is simple and has been exploited before, such as in
\cite{Bar}).  The union of orthants is known to have worst-case combinatorial
complexity $O(n^{\down{D/2}})$~\cite{Boi}.  Interestingly, we show that it is possible to
maximize certain types of objective functions over the complement
of the union, in time significantly smaller than the worst-case combinatorial
complexity.

We accomplish this by adapting known techniques on 
Klee's measure problem \cite{OveYap,ChaSoCG08,Bri,ChaFOCS13}.
Specifically, we build on a remarkable method by Bringmann~\cite{Bri} for
computing the volume of a union
of $n$ orthants in $D$ dimensions in $O(n^{D/3+O(1)})$
time (the $O(1)$ term in the exponent was $2/3$ but has been later removed
by author~\cite{ChaFOCS13}).  However, maximizing an objective function over the complement of the union is different from
summing or integrating a function, and Bringmann's method does not
immediately generalize to the former (for example, it exploits subtraction).  We introduce extra ideas to extend the method,
which results in a bigger time bound than $n^{D/3}=n^{2d/3}$ but nevertheless
beats $n^{D/2}=n^d$.  In particular, we use some simple graph-theoretical
arguments, applied to graphs with $O(D)$ vertices.

\PAPER{
\paragraph{Paper organization.}
We present our 2D algorithm in Section~\ref{sec:2d}, our 3D algorithm in 
Section~\ref{sec:3d}, and 
our higher-dimensional algorithms in Sections \ref{sec:anchored}--\ref{sec:box}
(all these parts may be read independently).
}
\LIPICS{
\paragraph{Organization.}
We present our 2D algorithm in Sec.~\ref{sec:2d}, our 3D algorithm in 
the full paper, and 
our higher-dimensional algorithms in Sec.~\ref{sec:anchored}--\ref{sec:box}
(all these parts may be read independently).
}

%In the rest of the paper, the presentation of our
%2D algorithm (in Section~\ref{sec:2d}),
%3D algorithm (in Appendix~\ref{sec:3d}),
%and higher-dimensional algorithms (Sections \ref{sec:anchored}--\ref{sec:box}) may be read independently.
%%We end with remarks in Section~\ref{sec:rmks}.

\section{\PAPER{Largest Empty Rectangle in 2D}\LIPICS{Largest empty rectangle in 2D}}\label{sec:2d}

\newcommand{\RAY}{\makebox[0pt]{\raisebox{5.5pt}{\tiny\ \ \ $\leftarrow$}}\gamma} 

As in previous work~\cite{ChDrLe,AggSur}, we focus on
solving a \emph{line-restricted} version of the 2D largest empty rectangle problem: given a set
$P$ of $n$ points below a fixed horizontal line $\ell_0$ and a set $Q$ of $n$ points above $\ell_0$, where the $x$-coordinates of all points
have been pre-sorted, and given a rectangle $B_0$, find
the largest-area rectangle $B\subset B_0$ that intersects $\ell_0$ and is 
empty of points of $P\cup Q$.
By standard divide-and-conquer,
an $O(T(n))$-time algorithm for
the line-restricted problem immediately yields
an $O(T(n)\log n)$-time algorithm for the original largest empty rectangle
problem,
assuming that $T(n)/n$ is nondecreasing.
%(Note that previous algorithms by Chazelle, Drysdale, and Lee~\cite{ChDrLe} and Aggarwal and Suri~\cite{AggSur} also solved the line-restricted subproblem first, but with a vertical dividing line $\ell_0$;
%we have switched $x$- and $y$-coordinates for convenience.)

%Without loss of generality, we may assume that $\ell_0$ is the $x$-axis.

We begin by reformulating the line-restricted problem as a problem about horizontal line segments.  In the subsequent subsections, we will work with this re-formulation.

For each point $p\in P$, let $s(p)$ be the longest horizontal line segment
inside $B_0$ such that $s(p)$ passes through $p$ and
there are no points of $P$ above $s(p)$.  See Figure~\ref{fig:cart}(a).
We can compute $s(p)$ for all $p\in P$ in $O(n)$ time: this step is equivalent to the construction of the standard \emph{Cartesan tree}~\cite{Vui,GaBeTa}, for which there are simple linear-time algorithms
(for example, by 
inserting points
from left to right and maintaining a stack, like Graham's scan,
as also re-described in previous papers~\cite{ChDrLe,AggSur}).
Similarly, for each $q\in Q$, let $t(q)$ be the longest horizontal line segment inside $B_0$ such that $t(q)$ passes through $q$ and
there are no points of $Q$ below $t(q)$.
We can also compute $t(q)$ for all $q\in Q$ in $O(n)$ time.

% and $t(q)$ for all $q\in Q$ in $O(n)$ time
%As noted by Chazelle, Drysdale, and Lee~\cite{ChDrLe}, 
% by an algorithm resembling Graham's scan (in their terminology, this step corresponds to the computation of
%the ``left support'' and ``right support'' of every point; it is also
%equivalent to ).

\begin{figure}
\begin{center}
\includegraphics[scale=0.72]{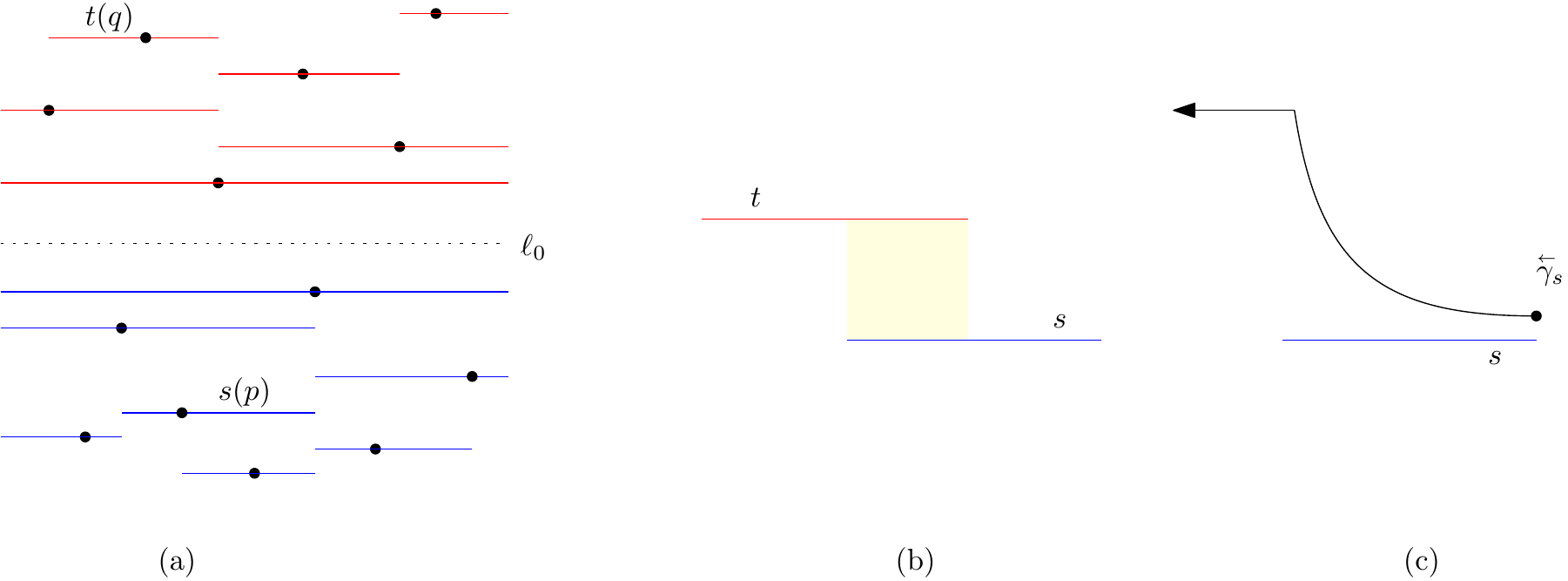}
\end{center}
\caption{(a,b) Transforming points into horizontal segments. (c) Pseudo-ray $\RAY_s$.}\label{fig:cart}
\end{figure}

For a horizontal segment $s$, let $x_s^-$ and $x_s^+$ denote the $x$-coordinates
of its left and right endpoints respectively, and let $y_s$ denote its $y$-coordinate.
% let $I_s$ denote the projected $x$-interval $[x_s^-,x_s^+]$.
We say that a set $S$ of horizontal segments is \emph{laminar} if for every $s,s'\in S$, either the two intervals
$[x_s^-,x_s^+]$ and $[x_{s'}^-,x_{s'}^+]$ are disjoint, or
one interval is contained in the other
(in other words, the intervals form a ``balanceded parentheses'' or tree structure).
It is easy to see that for the segments defined above,
$\{s(p):p\in P\}$ is laminar and $\{t(q):q\in Q\}$ is laminar.

The optimal rectangle must have some point $p^*\in P$ on its bottom side
and some point $q^*\in Q$ on its top side (except when
the optimal rectangle touches the bottom or top side of $B_0$, a case that can be easily dismissed in linear time).
% but
%the condition can be met by adding $O(n)$ extra points to $P$ and $Q$ on the bottom and top sides of $B_0$).
Chazelle, Drysdale, and Lee~\cite{ChDrLe} already noted that the case when
$[x^-_{s(p^*)},x^+_{s(p^*)}]$ is contained in $[x^-_{t(q^*)},x^+_{t(q^*)}]$ can be handled in $O(n)$ time
(in their terminology, this is the case of ``three supports
in one half, one in the other'').\footnote{
The solution is simple: for each $p\in P$, we find the lowest point $q_p\in Q$ with $x$-coordinate in the interval $[x^-_{s(p)},x^+_{s(p)}]$, and take the maximum of $(x^+_{s(p)}-x^-_{s(p)})(y_{t(q_p)} - y_{s(p)})$.
%rectangle with
%bottom side $s(p)$ and top side containing $q_p$; we return the best rectangle found.  
All these lowest points $q_p$ can be found ``bottom-up'' in the tree
formed by the intervals $\{[x^-_{s(p)},x^+_{s(p)}]: p\in P\}$, in linear total time.
}
The key remaining case is when
$x_{t(q^*)}^- < x_{s(p^*)}^- < x^+_{t(q^*)} < x^+_{s(p^*)}$,
where the area of the optimal rectangle is
$(x^+_{t(q^*)} - x_{s(p^*)}^-)(y_{t(q^*)}-y_{s(p^*)})$.
All other cases are symmetric.  
The problem is thus reduced to the following (see Figure~\ref{fig:cart}(b)):

\begin{problem}\label{prob1}
Given a laminar set $S$ of $n$ horizontal segments and a laminar set $T$ of $n$ horizontal segments, where all $x$-coordinates have been pre-sorted, find a pair $(s,t)\in S\times T$ such that
$x_t^- < x_s^- < x_t^+ < x_s^+$, maximizing
$(x_t^+ - x_s^-)(y_t-y_s)$.
\end{problem}

We find it more convenient to work with the corresponding \emph{decision problem}, as stated below.
By the author's randomized optimization technique~\cite{ChaSoCG98}, an $O(T(n))$-time
algorithm for Problem~\ref{prob2} yields an $O(T(n))$-expected-time
algorithm for Problem~\ref{prob1}, assuming that $T(n)/n$ is nondecreasing:

\begin{problem}\label{prob2}
Given a laminar set $S$ of $n$ horizontal segments and a laminar set $T$ of $n$ horizontal segments, where all $x$-coordinates have been pre-sorted,  and given a value $r>0$, 
decide if there exists a pair $(s,t)\in S\times T$ such that
$x_t^- < x_s^- < x_t^+ < x_s^+$ and 
$(x_t^+ - x_s^-)(y_t-y_s) > r$, and if so, report one such pair.  We call such a pair
\emph{good}.
\end{problem}

\subsection{Preliminaries}

%\newcommand{\cev}[1]{\reflectbox{\ensuremath{\vec{\reflectbox{\ensuremath{#1}}}}}}

%{\cev{\gamma}}%{\stackrel{\tiny\leftarrow}{\gamma}}
%{\vec{\gamma}}

To help solve Problem~\ref{prob2},
we define a curve $\gamma_s$ for each $s\in S$:
\[ \gamma_s(x) = \left\{\begin{array}{ll}
\frac{r}{x-x_s^-} + y_s & \mbox{if $x\ge x_s^-+\delta$}\\
M +x_s^-  & \mbox{if $x < x_s^-+\delta$,}
\end{array}\right.
\]
for a sufficiently
small $\delta>0$ and a sufficiently large $M=M(\delta)$.
(The main first part of the curve is a hyperbola.)
The condition $(x_t^+ - x_s^-)(y_t-y_s) > r$ is met iff the point
$(x_t^+,y_t)$ (i.e., the right
endpoint of $t$) is above the curve $\gamma_s$, assuming $x_t^+\ge x_s^- +\delta$.
Note that these curves form a family of \emph{pseudo-lines}, i.e., every pair
of curves intersect at most once: this can be seen from the fact that for any two curves $\gamma_s$ and $\gamma_{s'}$ with $x_s^-\ge x_{s'}^-$,
the difference
$\gamma_{s}(x)-\gamma_{s'}(x) = \tfrac{r(x_{s}^- - x_{s'}^-)}{(x-x_s^-)(x-x_{s'}^-)} + y_s - y_{s'}$ is nonincreasing for $x\ge x_s^-$.

Define the curve segment $\RAY_s$ to be the part of $\gamma_s$ restricted
to $x\le x_s^+$.  (See Figure~\ref{fig:cart}(c).)  These curve segments form a family of \emph{pseudo-rays}.  The lower envelope of $n$ pseudo-rays has at most $2n$ edges,
by known combinatorial bounds on order-2 Davenport-Schinzel sequences~\cite{ShaAgaBOOK}.  The following lemma summarizes known subroutines we  need on the computation of lower envelopes (proofs are briefly sketched).

%We can extend these curve segments $\gamma_s$ to \emph{pseudo-rays}, i.e., pseudo-line-segments with left endpoint at $x=-\infty$, by 
%defining $\hat{\gamma}_s$ as follows:
%$\hat{\gamma}_s(x) = \gamma_s(x)$ for $x\in [x_s^-+\delta,x_s^+]$
%and $\hat{\gamma}_s(x) = M+x_s^-$ for $x < x_s^- + \delta$
%with a sufficiently large constant~$M$.  Each pair still intersects
%at most once.

%for each $q\in Q$, 
%$\gamma_q(x) = y_q - r/(x_q'-x)$ for $x_q\le x\le x_q'$;

\begin{lemma}\label{lem:LE}
Consider a set of $n$ pseudo-lines, sorted by their
\emph{pseudo-slopes}, such that if $\gamma$ and $\gamma'$ intersects and $\gamma$ has smaller pseudo-slope, then $\gamma$ is above $\gamma'$ to the left of the intersection.  Assume that the intersection of any two pseudo-lines
can be computed in constant time.
\begin{enumerate}
\item[\rm (a)] Consider $n$ pseudo-rays that are parts of the given pseudo-lines,
such that the $x$-coordinates of the left endpoints are all $-\infty$,
and the $x$-coordinates of the right endpoints
are monotone (increasing or decreasing) in the pseudo-slopes.
Then the lower envelope of these pseudo-rays can be computed in $O(n)$ time.
\item[\rm (b)]
Consider $n$ pseudo-segments that are parts of the given pseudo-lines,
such that $x$-coordinates of the left endpoints are monotone in the pseudo-slopes and the $x$-coordinates of the right endpoints are monotone in the pseudo-slopes.
Then the lower envelope of these pseudo-segments can be computed in $O(n)$ time.
\end{enumerate}
\end{lemma}
\begin{proof}
Part (a) follows by a straightfoward variant of Graham's scan~\cite{BerBOOK} (originally for computing planar convex hulls, or by
duality, lower envelopes of lines).  We insert pseudo-rays in decreasing order of their right endpoints' $x$-values, while maintaining
the portion of the lower envelope to the left of the right endpoint of
the current pseudo-ray.  
In each iteration, by the monotonicity assumption, 
a prefix or suffix of the lower envelope gets deleted (i.e., popped from a stack). 

For part (b), the main case is when both
the left and right endpoints are monotonically increasing in the pseudo-slopes (the case when both are monotonically decreasing is symmetric,
and the case when they are monotone in different directions easily
reduces to two instances of the pseudo-ray case).
Greedily construct a minimal set of vertical lines that stab all 
the pseudo-segments: namely, draw a vertical line at the leftmost right endpoint, remove all pseudo-segments stabbed, and repeat.  This process
can be done in $O(n)$ time by a linear scan.  These vertical lines
divide the plane into slabs.  (See Figure~\ref{fig:fifo}.)  In each slab, the pseudo-segments
behave like pseudo-rays, so we can compute the lower envelope inside
the slab in linear time by applying part (a) twice, for the leftward rays and for the rightward
rays (the two envelopes can be merged in linear time).
Since each pseudo-segment participates in at most two slabs, the total time is linear. 
\end{proof}

\begin{figure}
\begin{center}
\includegraphics[scale=0.6]{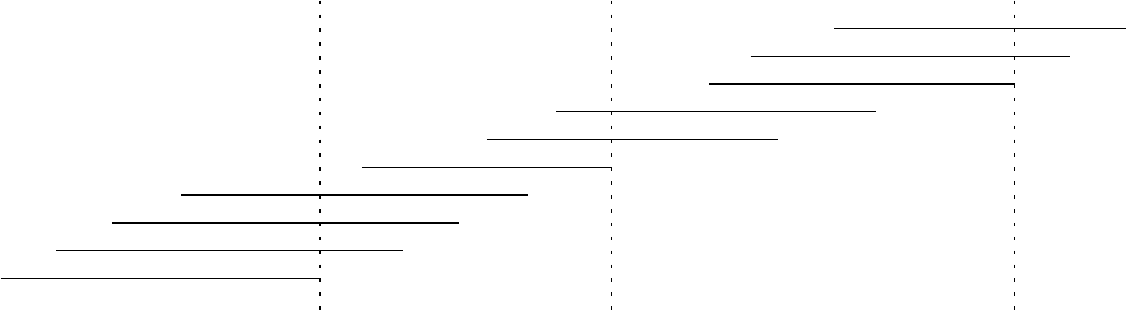}
\end{center}
\caption{Proof of Lemma~\ref{lem:LE}(b): the $x$-projected intervals and the division into slabs.}\label{fig:fifo}
\end{figure}

As an application of Lemma~\ref{lem:LE}(b), we mention an efficient algorithm for
a special case of Problem~\ref{prob2}, which will be useful later.

\begin{corollary}\label{cor}
In the case when all segments in $S$ and $T$ intersect a fixed vertical line, Problem~\ref{prob2} can be solved in $O(n)$ time.
\end{corollary}
\begin{proof}
Since $S$ and $T$ are laminar, the $x$-projected intervals in each set 
are nested.  Let $s_1,s_2,\ldots$ be the segments in $S$ with
$[x^-_{s_1},x^+_{s_1}]\subseteq [x^-_{s_2},x^+_{s_2}]\subseteq\cdots$, and
let $t_1,t_2,\ldots$ be the segments in $T$ with
$[x^-_{t_1},x^+_{t_1}]\subseteq [x^-_{t_2},x^+_{t_2}]\subseteq\cdots$.
For each $s_i$, let $a(i)$ be the smallest index with $x_{t_{a(i)}}^- <  x_{s_i}^-$, let $b(i)$ be the smallest index with $x_{s_i}^- < x_{t_{b(i)}}^+$, and let $c(i)$ be the largest index with
$x_{t_{c(i)}}^+ < x_{s_i}^+$.
Note that $a(i)$ is monotonically increasing in $i$, and
$b(i)$ is monotonically decreasing in $i$,
and $c(i)$ is monotonically increasing in $i$.
It is straightforward to compute $a(i),b(i),c(i)$ for all $i$ by a linear scan.

\newcommand{\SEG}{\overline{\gamma}}
The problem reduces to finding a pair $(s_i,t_j)$ such that
$\max\{a(i),b(i)\}\le j\le c(i)$ and the right endpoint of $t_j$
is above $\gamma_{s_i}$.
Define the curve segment $\SEG_{s_i}$ to be the part of $\gamma_{s_i}$ restricted
to $x\in [\max\{x_{t_{a(i)}}^+,x_{t_{b(i)}}^+\}, x_{t_{c(i)}}^+]$.
The problem reduces to finding a $t_j$ whose right endpoint is above
some curve segment $\SEG_{s_i}$, i.e., above the lower envelope of these curve segments.
We can compute this lower envelope 
in $O(n)$ time by Lemma~\ref{lem:LE}(b) (more precisely,
by two invocations of the lemma, as $\max\{x_{t_{a(i)}}^+,x_{t_{b(i)}}^+\}$
consists of a monotonically increasing and a monotonically decreasing part).  The problem can be then be solved by linear scan over the envelope and the
endpoints of $t_j$.
\end{proof}

\IGNORE{
\paragraph{Remarks.}
The above subroutines and their relevance to largest empty rectangle
are not exactly new.  The previous work by Aggarwal and Suri~\cite{AggSur}
exploited techniques for matrix searching, namely, finding row minima in certain types of partial Monge matrices.  The row minima problem is 
related to the computation of lower envelopes of pseudo-rays and pseudo-segments~\cite{??}.
Lemma~\ref{lem:LE}(a) can be reduced to row minima
staircase Monge matrices, and Lemma~\ref{lem:LE}(b) and Corollary~\ref{cor} can be reduced to row minima in double-staircase Monge matrices
(the preceding proof for Lemma~\ref{lem:LE}(b) is equivalent to Aggarwal and Klawe's reduction
of double-staircase to staircase matrices~\cite{AggKla}).
Klawe and Kleitman~\cite{KlaKle} gave an $O(n\alpha(n))$-time algorithm
for row minima in staircase Monge matrices, which has recently been improved to $O(n)$ by the author~\cite{??}, but these more complicated
subroutines are not necessary in our application: because we have re-directed effort to solving
the \emph{decision problem}, we only need to deal with pseudo-lines or
pseudo-segments of
\emph{constant} complexity (defined by hyperbolas), for which there are more direct algorithms for
constructing lower envelopes (as we have explained).
}

\subsection{Algorithm}

We are now ready to describe our new algorithm for solving Problem~\ref{prob2}, using interval trees and an interesting recursion with $O(\log^*n)$ depth.
%We use the well-known ``interval trees'' 
%to divide the problem into
%smaller subproblems of logarithmic size, yielding
%iterated logarithmic complexity.

\begin{theorem}
Problem~\ref{prob2} can be solved in $O(n2^{O(\log^*n)})$ time.
\end{theorem}
\begin{proof}
As a first step, we build the standard \emph{interval tree} for the given horizontal segments in $S\cup T$.
This is a perfectly balanced binary tree of with $O(\log n)$ levels,
where each node corresponds to a vertical slab.
The root slab is the entire plane, the slab at a node is the
union of the slabs of its two children, and each leaf slab contains
no endpoints in its interior.  Each segment is stored in the lowest node $v$ whose slab contains the segment (i.e., the segment is contained in $v$'s slab but
is not contained in either child's subslab).  Note that each segment is
stored only once (unlike in another standard structure called the ``segment tree'').
We can determine the slab containing each segment in $O(1)$ time by
an LCA query~\cite{BenFar} (which is easier in the case of a perfectly balanced binary tree).

\newcommand{\LE}{{\cal E}}

For each node $v$, let $S_v$ (resp.\ $T_v$)  be the set of all segments of $S$ (resp.\ $T$) 
stored in $v$. % and $T_v$ be the set of all segments of $T$ stored in~$v$.
Define the \emph{level} of a segment to be the level of the node it is stored in.

\bigskip
\noindent{\sc Case 1:} there exists a good pair $(s^*,t^*)$ where $s^*$ and $t^*$ have the same level.
Here, $s^*$ and $t^*$ must be stored in the same node $v$ of the interval tree.  Thus, a good pair can be found as follows:

\begin{enumerate}
\item[1.]
For each node $v$, solve the problem for $S_v$ and $T_v$ by 
Corollary~\ref{cor} in $O(|S_v|+|T_v|)$ time.  Note that all segments in $S_v\cup T_v$ indeed
intersect a fixed vertical line (the dividing line at $v$).
\end{enumerate}

\noindent
The total running time of this step is $O(n)$,
since each segment is in only one $S_v$ or $T_v$.

\bigskip\noindent{\sc Case 2:} there exists a good pair $(s^*,t^*)$ where $s^*$ is on a strictly lower level than $t^*$.
To deal with this case, we perform the following steps, for some choice of parameter $b\ge \log n$:

\begin{enumerate}
\item[2a.]
For each node $v$, compute the lower envelope of the pseudo-rays
$\{\RAY_s: s\in S_v\}$ by Lemma~\ref{lem:LE}(a) in $O(|S_v|)$ time;
let $\LE_v$ denote this envelope restricted to $v$'s slab.
Note that because all segments in $S_v$ intersect a fixed vertical line
and $S_v$ is laminar, the $x_s^+$ values are
monotonically decreasing in the $x_s^-$ values for $s\in S_v$ and
so are indeed monotone in the pseudo-slopes of these pseudo-rays.
%and the upper envelope $\UE_v$ of $\{\gamma_q: q\in Q_v\}$

%Let $\LE_i$ be the concatenation of the lower envelopes $\LE_v$ (or more precisely, the part of $\LE_v$ inside $v$'s slab) over all nodes $v$ at level $i$ (the slabs are disjoint for nodes at the same level).  
%The total number of edges in all these envelopes is
%at most $2\sum_v |S_v|\le 2n$.
\item[2b.]
Divide the plane into a set $\Sigma$ of $n/b$ vertical slabs each
containing $b$ right endpoints of $T$.
\item[2c.]
For each slab $\sigma\in \Sigma$,
\begin{itemize}
\item
let $T_\sigma$ be the set of all segments $t\in T$ with right endpoints in $\sigma$, and 
\item let $S_\sigma$ be the set of all segments $s\in S$
such that $\RAY_s$ appears on $\LE_v\cap\sigma$ for some node~$v$.
\end{itemize}
Divide $S_\sigma$ (arbitrarily) into blocks of size $b$ and
recursively solve the problem for $T_\sigma$ and each block of~$S_\sigma$.
%Note that $|T_\sigma|\le b$ and $|S_\sigma|\le O(b\log n)$ for each $\sigma\in \Sigma$.
%Clearly, $\sum_{\sigma\in\Sigma} |T_\sigma| \le n$.
%Also, since the dividing vertical lines of $\Sigma$ intersect
%at most $O((n/b)\log n)$ edges among all the envelopes $\LE_i$,
%we have
%$\sum_{\sigma\in\Sigma} |S_\sigma| \le 2n + O((n/b)\log n)$.
\end{enumerate}

\medskip\noindent
\emph{Correctness.} 
Consider a good pair $(s^*,t^*)$
 with $s^*$ on a strictly lower level than $t^*$.
Let $\sigma$ be the slab in $\Sigma$ containing the
right endpoint of $t^*$, i.e., $t^*\in T_\sigma$.
Let $v$ be the node $s^*$ is stored in.
Then $t^*$ intersects the left wall of the slab at $v$
(since $t^*$ must be stored in a proper ancestor of $v$).
Now, the right endpoint of $t^*$ is below $\RAY_{s^*}$ and
is thus below $\LE_v$.
Let $\RAY_{s}$ be the curve on $\LE_v$ that the right endpoint of $t^*$ is below, with $s\in S_v$.
Then $\RAY_s$ appears on $\LE_v\cap\sigma$, and
so $s\in S_\sigma$.
Since the right endpoint of $t^*$ is below $\RAY_s$,
we have $x_{s}^- < x_{t^*}^+ < x_{s}^+$,
and since $t^*$ intersects the left wall of $v$'s slab,
we have $x_{t^*}^- < x_{s}^-$.
So, $(s,t^*)$ is good, and the
recursive call for $T_\sigma$ and some block of $S_\sigma$ 
will find a good pair.

\medskip\noindent
\emph{Analysis.} 
The total number of edges in all envelopes $\LE_v$ is
at most $2\sum_v |S_v|\le 2n$.
Since the envelopes $\LE_v$ have disjoint $x$-projections for nodes $v$
at the same level, and since there are $O(\log n)$ levels,
the $O(n/b)$ dividing vertical lines of $\Sigma$ intersect
at most $O((n/b)\log n)$ edges among all the envelopes.
Thus, $\sum_{\sigma\in\Sigma} |S_\sigma| \le 2n + O((n/b)\log n)=O(n)$ if $b\ge\log n$,
and so the total number of recursive calls in step 2c is
$O(n/b)$.

\bigskip\noindent
{\sc Case 3:} there exists a good pair $(s^*,t^*)$ where $s^*$ is on a strictly higher level than $t^*$.
This remaining case is symmetric to Case 2 (by switching $S$ and $T$ and negating
 $y$-coordinates).

\bigskip
By running the algorithms for all three cases, a good pair is guaranteed to be found if one exists.  
The running time satisfies the recurrence
$T(n) \le O(n/b)T(b) + O(n)$.  Setting $b=\log n$ gives
%, we obtain an algorithm for Problem~\ref{prob2} with time bound 
$T(n)\le n2^{O(\log^*n)}$.
\end{proof}

By the observations from the beginning of this section, we can now solve Problem~\ref{prob1}
and the line-restricted problem in $O(n2^{O(\log^*n)})$ expected time,
and the original largest empty rectangle problem in 
$O(n2^{O(\log^*n)}\log n)$ expected time.

\begin{corollary}
Given $n$ points in $\R^2$ and a rectangle $B_0$,
we can compute the maximum-area empty rectangle inside $B_0$ in
$O(n2^{O(\log^*n)}\log n)$ expected time.
\end{corollary}

\PAPER{%%%%%%%%

\section{Largest Empty Box in 3D}\label{sec:3d}

In this section, we describe a subcubic algorithm for the largest empty box problem
in 3D\@.  The key is the following result on an ``asymmetric'' case with a left point set and right point set of different sizes:

\newcommand{\xx}{x}

\begin{theorem}
Given a set $P$ of $n$ points in $(-\infty,0)\times \R^2$, and
a set $Q$ of $m$ points in $(0,\infty)\times \R^2$,
we can compute the maximum-volume box that
contains the origin and
is empty of points in $P\cup Q$ in $\OO(n^2 + m^{4+\eps})$ time
for an arbitrarily small constant $\eps>0$.
\end{theorem}
\begin{proof}
Map a box $b=(-x_1,x_1')\times (-x_2,x_2')\times
(-x_3,x_3')\subset \R^3$
to a point $b^*= (x_1,x_1',x_2,x_2',x_3,x_3')$ 
in 6D\@.
Map a point $p=(p_1,p_2,p_3)\in \R^3$
to an orthant $p^* = (-p_1,\infty)\times (p_1,\infty)\times
(-p_2,\infty)\times (p_2,\infty)\times (-p_3,\infty)\times (p_3,\infty)$
in 6D\@.
Then the point $p$ is in the box $b$ iff the point $b^*$ is in the orthant~$p^*$.

\newcommand{\XX}{{\cal X}}
Say $B_0=(-c_1,c_1')\times (-c_2,c_2)\times (-c_3,c_3)$.
Define $\XX = [0,c_1]\times [0,c_1']\times [0,c_2]\times [0,c_2']\times [0,c_3]\times [0,c_3']$.  
%Then a box $b\subseteq \D_0$ containing
%the origin maps to a point $b^*$ in $\XX$, and vice versa.

Our goal is to find a point $b^*=(x_1,x_1',x_2,x_2',x_3,x_3')\in \XX$ maximizing the function $H(b^*):=(x_1+x_1')(x_2+x_2')(x_3+x_3')$,
such that $b^*$ is in the complement of both $U(P)=\bigcup_{p\in P} p^*$
and $U(Q)=\bigcup_{q\in Q} q^*$.  Let $Z$ be the set of all vertices of $U(P)\cap\XX$ and $A$ be the set of all vertices of $U(Q)\cap \XX$.
% (points at infinity are included in $Z$ and $A$).  
The constraint can then be restated as follows: $b^*$ is dominated
by some vertex in $Z$ and by some vertex in $A$.  

Since all points $(p_1,p_2,p_3)\in P$
have $p_1<0$, $U(P)\cap \XX$ corresponds to a union of $n$ orthants in 5D (the second dimension is irrelevant); by known
results on the union complexity of orthants~\cite{Boi},
the set $Z$ has $O(n^2)$ size and can be constructed in $\OO(n^2)$ time.
Similarly, since all points $(q_1,q_2,q_3)\in Q$
have $q_1>0$, $U(Q)\cap \XX$ corresponds to a union of $n$ orthants in 5D (the first dimension is irrelevant);
the set $A$ has $O(m^2)$ size and can be constructed in $\OO(m^2)$ time.

For each $z=(z_1,z_1',z_2,z_2',z_3,z_3')\in Z$
and each $a=(a_1,a_1',a_2,a_2',a_3,a_3')\in A$, the maximum of $h(b^*)$ over all $b^*\in \XX$ that 
are dominated by both $z$ and $a$ is
given by the function
\[ f(z,a) \: :=\: (\min\{z_1,a_1\} + \min\{z_1',a_1'\})\cdot
(\min\{z_2,a_2\} + \min\{z_2',a_2'\})\cdot
(\min\{z_3,a_3\} + \min\{z_3',a_3'\}).
\]
Thus, our problem is reduced to computing the maximum of $f(z,a)$
over all $z\in Z$ and $a\in A$.

This new problem can be solved by standard range searching technique.
First consider the decision problem of testing whether the maximum
exceeds a given fixed value $r$.
We build a two-level data structure for $A$ and ask a query for each $z\in Z$,
to decide whether there exists an $a\in A$ with $f(z,a)\ge r$:

For each $z\in Z$, in one case, we identify all $a\in A$ with
$z_1\ge a_1$, $z_1'\le a_1'$,
$z_2\ge a_2$, $z_2'\le a_2'$, $z_3\ge a_3$, and $z_3'\le a_3'$,
by orthogonal range searching (a range tree)~\cite{BerBOOK}.
The answer can be expressed as a union of $\OO(1)$ canonical subsets $A_i$.
For each such canonical subset $A_i$,
we decide whether there exists an $a\in A_i$ with
$(z_1 + a_1')(z_2 + a_2')(z_3 + a_3') \ge r$.
This can be done by point location in a lower envelope
of surfaces of the form $z_3 = \frac{r}{(z_1+a_1')(z_2+a_2')} - a_3'$ in 3D\@.
By known results on lower envelopes of surfaces~\cite{Sha94}, with $O(|A_i|^{2+\eps})$-time preprocessing, a query
can be answered in $O(\log m)$ time.

In another case, for example, when $z_1\ge a_1$, $z_1'\ge a_1'$,
$z_2\ge a_2$, $z_2'\le a_2'$, $z_3\ge a_3$, and $z_3'\le a_3'$,
we decide whether there exists an $a\in A_i$ with
$(z_1 + z_1')(z_2 + a_2')(z_3 + a_3') \ge r$.
With a change of variable $z_1''=z_1+z'$, this can be
done by point location in a lower envelope of surfaces of the
form $z_3 = \frac{r}{z_1''(z_2+a_2)}-a_3'$ in 3D\@.
Other cases are similar (or easier).

The entire two-level data structure
has $O(|A|^{2+\eps})$ preprocessing time and $\OO(1)$ query time.
The total time for $|Z|$ queries is $\OO(|A|^{2+\eps}+|Z|) = \OO(m^{4+O(\eps)} + n^2)$. 

The original problem can be reduced to the decision
problem, for example, by the author's randomized optimization technique~\cite{ChaSoCG98}, or deterministically, by parametric search~\cite{Meg}
(since the preprocessing algorithm can be parallelized).
\end{proof}

\begin{corollary}
Given $n$ points in $\R^3$ and a box $B_0$,
we can compute the maximum-volume empty box inside $B_0$ in $O(n^{5/2+\eps})$ time
for an arbitrarily small constant $\eps>0$.
\end{corollary}
\begin{proof}
By divide-and-conquer,
it suffices to solve the \emph{plane-restricted} version where the
box is constrained to intersect a given axis-parallel plane.
By another application of divide-and-conquer,
the problem can be further reduced to the \emph{line-restricted}
version when the box is constrained to intersect a given axis-parallel line $\ell_0$.  The running time increases by at most
two logarithmic factors.
Without loss of generality, assume that $\ell_0$
is the first coordinate axis.

Divide space into $n/m$ slabs, by planes 
orthogonal to the first coordinate axis, where each slab contains
$m$ points.  Order the slabs from left to right.  Let $Q_i$ be the subset of all input points in the $i$-th slab $\sigma_i$.

Consider the case when the optimal box has its
right side in $\sigma_i$ but is not contained in $\sigma_i$.
This case reduces to an instance of the above lemma for the two point sets $Q_1\cup\cdots \cup Q_{i-1}$ and $Q_i$ (after translation to make
the left wall of $\sigma_i$ pass through the origin).
The total cost over all $n/m$ slabs $\sigma_i$ is thus 
$\OO((n/m)\cdot (n^2 + m^{4+\eps}))$.
Setting $m=\sqrt{n}$ gives a time bound of $O(n^{5/2+\eps})$.

The remaining case when the optimal box is contained in $\sigma_i$ for some $i$
can be handled by recursion.  The total time is
$T(n)=\sqrt{n}T(\sqrt{n}) + O(n^{5/2+\eps})$, yielding $T(n)=O(n^{5/2+\eps})$.
(Alternatively, instead of recursion, we could switch to some known cubic algorithm.)
\end{proof}

In higher dimensions $d\ge 4$, a similar approach can yield an algorithm with running time of the
form $O(n^{d-1+O(1/d)})$, but the approach in Section~\ref{sec:box}
is better.  On the other hand, for $d=3$, the algorithm in Section~\ref{sec:box} gives time bound
$\OO(n^{(5d+2)/6})=\OO(n^{17/6})$, which is worse than $n^{5/2}$.

}%%%%%%%%%%%%%%

\section{\PAPER{Largest Empty Anchored Box in Higher Dimensions\PAPER{\\}
(Warm-Up)}\LIPICS{Largest empty anchored box in higher dimensions\PAPER{\\}
(warm-up)}}\label{sec:anchored}

To prepare for our solution to the largest empty box problem in higher constant dimensions, we first investigate a simpler variant,
the \emph{largest empty anchored box} problem:
given a set $P$ of $n$ points in $\R^d$ and a fixed box $B_0$, find the largest-volume anchored box  in $B_0$ that
does not contain any points of $P$ in its interior, where an
\emph{anchored} box has the form
$B=(0,x_1)\times\cdots (0,x_d)$ (having the origin as one of its vertices).

Let $\bigcup S$ denote the union of a set $S$ of objects.
By mapping a box $B=(0,x_1)\times\cdots (0,x_d)$ to the point $(x_1,\ldots,x_d)$,
and mapping each input point $(p_1,\ldots,p_d)$ to the
orthant $(p_1,\infty)\times\cdots\times (p_d,\infty)$, the largest
empty anchored box problem reduces to:% the following:

\newcommand{\Hvol}{H_{\mbox{\scriptsize\rm vol}}}
\newcommand{\Hnew}{H_{\mbox{\scriptsize\rm new-vol}}}

\begin{problem}\label{prob:anchored}
Define the function $\Hvol(x_1,\ldots,x_d)=x_1x_2\cdots x_d$.
Given a set $S$ of $n$ orthants in $\R^d$ and a box $B_0$,
find the maximum of $\Hvol$ over $B_0 - \bigcup S$.
\end{problem}

\PAPER{
(In the application to largest empty anchored box, the orthants all contain
$(\infty,\ldots,\infty)$, but our algorithm does not require all orthants
to be of the same type.)
}

By known results~\cite{Boi}, the union of $n$ orthants in $\R^d$ has
worst-case combinatorial complexity $O(n^{\down{d/2}})$ and can be
constructed in $\OO(n^{\down{d/2}})$ time.  We will
show that Problem~\ref{prob:anchored} can be solved faster than
explicitly constructing the union.

\subsection{Preliminaries}

A key tool we need is a spatial partitioning scheme due to Overmars and Yap~\cite{OveYap} (originally developed for solving Klee's measure problem
in $\OO(n^{d/2})$ time).
The version stated below is taken from \cite[Lemma~4.6]{ChaFOCS13}; see
that paper for a short proof.  (The partitioning scheme is also related to ``orthogonal BSP trees'' \cite{DuMiSh,ChaLee}.)

\begin{lemma}\label{lem:bsp}
Given a set of $n$ axis-parallel
flats (of possibly different dimensions)
in $\R^d$, and given a parameter $r$,
we can divide $\R^d$ into $O(r^d)$ cells (bounded and unbounded boxes) so that
each cell intersects $O(n/r^{j})$ $(d-j)$-flats.
% where
%$n_{\le j}$ is the number of flats of dimension at most $j$ in the
%given set.

The construction of the cells, along with the conflict lists
(lists of all flats intersecting each cell),
can be done in $\OO(n+r^d+K)$ time,\footnote{
A weaker time bound was stated in \cite[Lemma~4.6]{ChaFOCS13}, but
the output-sensitive time bound follows directly from the
same construction.
%easily by using
%standard data structures for orthogonal range searching
%or range intersection~\cite{AgaEriSURV,BerBOOK}.
}
where $K$ is the total size of the conflict lists.
\end{lemma}

Call a function $H:\R^d\rightarrow \R$ \emph{simple} if it has the form
\[ H(x_1,\ldots, x_d)\ =\  h_1(x_1)\cdots h_d(x_d),\]
where each $h_i$ is a univariate step function.
The \emph{complexity} of $H$ refers to the total complexity (number of steps) in these step functions. As an illustration of the usefulness of Lemma~\ref{lem:bsp},
we first how to maximize simple functions over the complement of a union of orthants in $\OO(n^{d/2})$ time:

\begin{lemma}\label{lem1:anchored}
Let $H$ be a simple function with $O(n)$ complexity.
Given a set $S$ of $n$ orthants in $\R^d$ and a box $B_0$,
we can compute the maximum
of $H$ in $B_0-\bigcup S$ in $\OO(n^{d/2})$ time  for any constant $d\ge 2$.
\end{lemma}
\begin{proof}
Apply Lemma~\ref{lem:bsp} to the $O(n)$ $(d-2)$-flats that pass through the $(d-2)$-faces of the given orthants.  This yields a partition of $B_0$
into cells.  

Consider a cell $\D$.  The number of $(d-2)$-flats intersecting $\D$
is bounded by $O(n/r^2)$, which can be made 0 by setting
$r:=\Theta(\sqrt{n})$.  Consequently, only $(d-1)$-faces of the given 
orthants may intersect $\D$, i.e., all orthants are 1-sided inside~$\D$.
The union of 1-sided orthants simplifies to the complement of a box
(we can use orthogonal range searching or intersection data structure to identify the 1-sided orthants intersecting $\D$ and compute this box
in $\OO(1)$ time~\cite{AgaEriSURV,BerBOOK}). 
For a simple function $H(x_1,\ldots,x_d)=h_1(x_1)\cdots h_d(x_d)$, 
we can maximize $H$ over a box by maximizing $h_i(x_i)$ over an interval
for each $i\in\{1,\ldots,d\}$ separately.
This corresponds to a 1D range maximum query for each $i$,
which can be done straightforwardly in $O(\log n)$ time (or more
carefully in $O(1)$ time~\cite{BenFar}).
As the number of cells is $O(r^d)=O(n^{d/2})$,
the total running time is $\OO(n^{d/2})$.
%plus the construction time from
%Lemma~\ref{lem:bsp}, which is $O(nr^d)=O(n^{d/2+1})$.
%We can do slightly better by choosing a smaller $r$ and using recursion:
%the running time satisfies the recurrence $T(n) = O(r^d) T(n/r^2) + O(nr^{d})$, with 
%$T(1)=\OO(1)$.  Setting $r=n^\delta$ for a sufficiently small $\delta>0$
%gives $T(n)=\OO(n^{d/2})$.
\end{proof}

\subsection{Algorithm}

To improve over $n^{d/2}$, we adapt an approach by Bringmann~\cite{Bri}
(originally for solving Klee's measure problem for orthants in $O(n^{d/3+O(1)})$ time).
The approach involves first solving the 2-sided special case, and then applying Overmars and Yap's partitioning scheme.  A \emph{2-sided
orthant} is the set of all points $(x_1,\ldots,x_d)\in\R^d$ satisfying
a condition of the form $[x_i\ ?\ a]\wedge [x_j\ ?\ b]$ for some $i,j\in\{1,\ldots,d\}$, where each occurrence of ``?''\ is either $\le$ or $\ge$.  We will adapt the author's subsequent re-interpretation~\cite[Section~4.1]{ChaFOCS13} of Bringman's 
technique, described in terms of monotone step functions.  

%(Bringmann~\cite{Bri} was able to solve the 2-sided case of Klee's measure problem in linear time, but maximizing a function appears
%more challenging than summing or integrating, as ``subtraction tricks'' no longer work.)

\begin{figure}
\begin{center}
\includegraphics[scale=0.6]{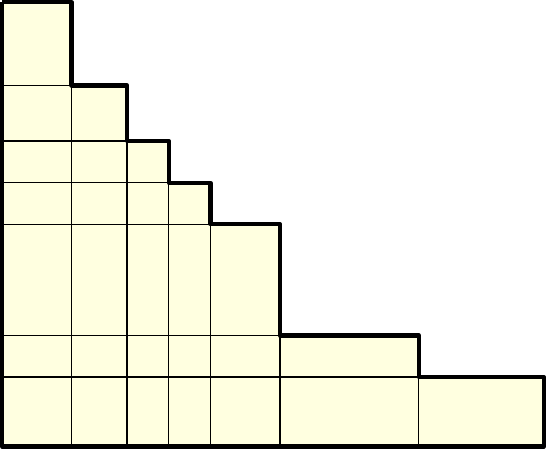}
\end{center}
\caption{The union of one type of 2-sided orthants.}\label{fig:staircase}
\end{figure}

\begin{theorem}\label{thm0:anchored}
In the case when all the input orthants are 2-sided,
Problem~\ref{prob:anchored} can be solved in $\OO(n^{\down{d/2}/2})$ time
for 
any constant $d\ge 4$.
\end{theorem}
\begin{proof}
The boundary of the union of 2-sided orthants of the form $[x_i\ ?\ a]\wedge [x_j\ ?\ b]$
with a fixed $i,j$ and a fixed choice for the two ``?''s
is a staircase, i.e., the graph
of a univariate monotone (increasing or decreasing) step function. 
(See Figure~\ref{fig:staircase}.)
Thus, the complement of the union of 2-sided orthants
can be expressed as the set of all points $(x_1,\ldots,x_d)\in\R^d$ satisfying an expression $E(x_1,\ldots,x_d)$ which is a conjunction of $O(d^2)$ predicates each of the form
$[x_i\ ?\ f(x_j)]$, where  $i,j\in\{1,\ldots,d\}$, ``?''\ is $\le$ or $\ge$, 
 and $f$ is a monotone step function.  The total
complexity of these step functions is $O(n)$.
Conversely, any such expression can be mapped back to the complement
of a union of $O(n)$ 2-sided orthants.

We first observe a few simple rules for rewriting expressions:

\begin{enumerate}
\item
$[x_i \le f(x_j)]\wedge [x_i \le g(x_j)]$ can be rewritten
as $[x_i\le \min\{f,g\}(x_j)]$ if $f$ and $g$ are both increasing or
both decreasing.  Note that the lower envelope $\min\{f,g\}$ 
is still a monotone step function with $O(n)$ complexity. 
A similar rule applies for $\ge$. 
%Thus, it suffices to keep at most four predicates
%for each of the $O(d^2)$ pairs $(i,j)$.
\item
$[x_i\le f(x_j)]$ can be rewritten as $[x_j\ge f^{-1}(x_i)]$
if $f$ is increasing (the inequality is flipped if $f$ is decreasing).
Note that the inverse $f^{-1}$ is still a monotone step function.
\item
More generally,
$[f(x_i)\le g(x_j)]$ can be rewritten as $[x_j\ge (f^{-1}\circ g)(x_i)]$
if $f$ is increasing (the inequality is flipped if $f$ is decreasing).  Note
that the composition $f^{-1}\circ g$ is still a monotone step
function with $O(n)$ complexity.
\item
$[x_i\le f(x_j)]\wedge [x_i\le g(x_k)]$ can be rewritten as 
the disjunction of 
$[x_i\le f(x_j)] \wedge [f(x_j)\le g(x_k)]$
and $[x_i\le g(x_k)] \wedge [g(x_k)\le f(x_j)]$.
A similar rule applies for $\ge$.
%Each of the two subexpressions can be cleaned up by applying rules 1--3.
\end{enumerate}

The plan is to decrease the dimension by repeatedly eliminating variables: %from the expression~$E$:

We maintain a simple function $H$.
Initially, $H(x_1,\ldots,x_d)=\sigma(x_1)\cdots\sigma(x_d)$,
where $\sigma(x)$ denotes the successor of $x$ among the $O(n)$ input
coordinate values ($\sigma$ is a step function).
% initially.  At any time, $H$ will be a simple function.
We call an index $i$ \emph{free} if the variable $x_i$ appears
exactly once in $H$ and is ``unaltered'', i.e., $h_i(x_i)=\sigma(x_i)$.  All indices are initially free.

In each iteration, we pick a free index $i$.
Whenever $x_i$ appears more than twice in $E$,
we can apply rule~4 (in combination with rules 1--3) 
to obtain a disjunction of 2 subexpressions,
where in each subexpression, the number of occurrences of $x_i$ is decreased.
By repeating this process $O(1)$ times (recall that $d$ is a 
constant), we obtain a disjunction of $O(1)$
subexpressions, where in each subexpression, only at most two occurrences of $x_i$ remain---in at most one predicate of the form
$[x_i\le f(x_j)]$, and at most one predicate of the form $[x_i\ge g(x_k)]$.
%(Note that after each application of rule~4, we use rules 1--3 to
%clean up the expressions.)

We branch off to maximize $H$ over each of these subexpressions separately.
In such a subexpression, to eliminate the variable $x_i$ while maximizing $H$,
we replace the two predicates $[x_i\le f(x_j)]$ and $[x_i\ge g(x_k)]$ with $[f(x_j)\ge g(x_k)]$,
and replace $x_i$ with $f(x_j)$ in $H$ (i.e., reset $h_j(x_j)$
to $h_j(x_j)\sigma(f(x_j))$, which is still a step function with $O(n)$ complexity).  Now, $i$ and $j$ are not free.

We stop a branch when there are no free indices left.
At the end, we get a large but $O(1)$ number of subproblems,
where in each subproblem,
at least $\up{d/2}$ variables have been eliminated, i.e., the 
dimension is decreased to $d'\le\down{d/2}$.
We solve each subproblem by Lemma~\ref{lem1:anchored}
in $\OO(n^{d'/2})$ time.
\end{proof}

We now combine Theorem~\ref{thm0:anchored} and Lemma~\ref{lem:bsp}
to solve Problem~\ref{prob:anchored}:

\begin{corollary}\label{thm:anchored}
Problem~\ref{prob:anchored} can be solved in $\OO(n^{d/3+\down{d/2}/6})$ time for any constant $d\ge 4$.
\end{corollary}
\begin{proof}
Apply Lemma~\ref{lem:bsp} to the $O(n)$ $(d-3)$-flats and $(d-2)$-flats through the $(d-3)$-faces and $(d-2)$-faces of
the given orthants.  This yields a partition of $B_0$ into cells.

Consider a cell $\D$.  The number of $(d-3)$-flats intersecting $\D$
is $O(n/r^3)$, which can be made 0 by setting
$r:=\Theta(n^{1/3})$.  The number of $(d-2)$-flats intersecting $\D$
is $O(n/r^2)=O(n^{1/3})$.  So, inside the cell $\D$, all orthants
are 2-sided or 1-sided, with $O(n^{1/3})$ 2-sided
orthants.  The union of 1-sided orthants simplifies to the complement
of a box (we can use orthogonal range searching or intersection data structures~\cite{AgaEriSURV,BerBOOK} to identify the 1-sided orthants intersecting $\D$ and compute this box).  We can thus apply Theorem~\ref{thm0:anchored} to maximize
$H$ over the cell $\D$ in $\OO((n^{1/3})^{\down{d/2}/2})$ time.
As there are $O(r^d)=O(n^{d/3})$ cells, the total running time is
$\OO(n^{d/3}\cdot (n^{1/3})^{\down{d/2}/2})$.
%plus the construction time
%from Lemma~\ref{lem:bsp} of $O(nr^d)=O(n^{d/3+1})$.
\end{proof}

\begin{corollary}
Given $n$ points in $\R^d$ and a box $B_0$,
we can compute the maximum-volume empty anchored box inside $B_0$ in $\OO(n^{d/3+\down{d/2}/6})\le \OO(n^{5d/12})$ time for any constant $d\ge 4$.
\end{corollary}

\section{\PAPER{Largest Empty Box in Higher Dimensions}%
\LIPICS{Largest empty box in higher dimensions}}\label{sec:box}

We now adapt the approach from Section~\ref{sec:anchored} to solve
the original largest empty box problem in higher constant dimensions.
By $d$ levels of divide-and-conquer, it suffices to solve
the \emph{point-restricted} version of the problem: given a set $P$
of $n$ points in $\R^d$, a fixed box $B_0$, and a fixed point $o$, find the largest-volume box $B\subset B_0$ that contains $o$ and is empty
of points of $P$.  An $O(T(n))$-time algorithm for the point-restricted problem immediately yields an $O(T(n)\log^dn)$-time algorithm for
the original problem (in fact, the polylogarithmic factor disappears if $T(n)/n^{1+\delta}$ is increasing for some constant $\delta>0$).
Without loss of generality, assume that $o$ is the origin.

By mapping a box $B=(-x_1,x_1')\times\cdots\times (-x_d,x_d')$ 
(which has volume $(x_1+x_1')\cdots (x_d+x_d')$) to 
the point $(x_1,x_1',\ldots,x_d,x_d')$ in $2d$ dimensions, 
and mapping each input point $p=(p_1,\ldots,p_d)$ to
the orthant $(-p_1,\infty)\times (p_1,\infty)\times\cdots
(-p_d,\infty)\times (p_d,\infty)$ in $2d$ dimensions (and changing $B_0$ appropriately),
the problem
reduces to the following variant of Problem~\ref{prob:anchored}, after doubling the dimension:

\begin{problem}\label{prob:box}
Define the function $\Hnew(x_1,\ldots,x_d)=(x_1+x_2)(x_3+x_4)\cdots (x_{d-1}+x_d)$ for an even $d$.
Given a set $S$ of $n$ orthants in $\R^d$ and a box $B_0$,
find the maximum of $\Hnew$ over $B_0 - \bigcup S$.
\end{problem}
 
The above objective function $\Hnew$ is a bit more complicated than
the one from Section~\ref{sec:anchored}, and so further ideas are needed\ldots

\subsection{Preliminaries}

For a multigraph $G$ with vertex set $\{1,\ldots,d\}$ (without self-loops),
define a \emph{$G$-function} $H:\R^d\rightarrow\R$ to be a function of the form
\[ H(x_1,\ldots, x_d)\ =\  \prod_{i=1}^d h_i(x_i)\cdot \prod_{e=ij\in G} (h'_e(x_i) + h''_e(x_j)),
\]
where $h_i$, $h'_e$, and $h''_e$ are univariate step functions.
The \emph{complexity} of $H$ refers to the total complexity of these step functions.

A \emph{pseudo-forest} is a graph where
each component is either a tree, or a tree plus an edge---in the latter
case, the component is called a \emph{1-tree} (and we allow the extra edge to be a duplicate of an edge in the tree).

\begin{lemma}\label{lem0}
Let $H$ be a $G$-function $H$ with $O(n)$ complexity.
Given a box $B_0$, we can compute the maximum of $H$ over $B_0$ in $\OO(n)$ time if $G$
is a forest, or $\OO(n^2)$ time if $G$ is a pseudo-forest, for any constant $d$.
\end{lemma}
\begin{proof}
For the forest case:
Pick a leaf $i$.
Then $H$ is of the form $h(x_i)\cdot (h'(x_i)+h''(x_j))\ \cdots$,
where $h,h',h''$ are step functions 
and $x_i$ does not appear in ``$\cdots$''.
Define $F(\xi) := \max_{x\in\R } h(x)\cdot (h'(x)+\xi)$.
Then $F$ is the upper envelope of $O(n)$ linear functions in the single
variable $\xi$, and can be constructed in $\OO(n)$ time by the dual of
a planar convex hull algorithm~\cite{BerBOOK}.
We can eliminate the variable $x_i$ by replacing
the $h(x_i)\cdot (h'(x_i)+h''(x_j))$ factor with $F(h''(x_j))$
(which is a step function in $x_j$ with $O(n)$ complexity).
As a result, $H$ becomes a $(G-\{i\})$-function in $d-1$ variables.
After $d$ iterations, the problem becomes trivial.

For the pseudo-forest case:
We may assume the graph is connected, since we can
maximize the parts of $H$ corresponding to different components separately.
Pick a vertex $i$ that belongs to the unique cycle of $G$ (if exists).
Then $G-\{i\}$ is a forest.
By trying out all $O(n)$ different settings of $x_i$ (breakpoints
of the step functions), 
the problem reduces to $O(n)$ instances of the forest case.
\end{proof}

\begin{lemma}\label{lem1}
Let $H$ be a 
$G$-function with $O(n)$ complexity, where $G$ is a pseudo-forest.
Given a set $S$ of $n$ boxes in $\R^d$ and a box $B_0$,
we can compute the maximum
of $H$ over $B_0-\bigcup S$ in $\OO(n^{d/2+1})$ time for any constant $d$. 
\end{lemma}
\begin{proof}
Apply Lemma~\ref{lem:bsp} to the $O(n)$ $(d-2)$-flats through the
boundaries of the orthants, together with the $O(n)$ $(d-1)$-flats $x_j=a$ for
all breakpoints $a$ of the step functions appearing in $H$.
This yields a partition of $B_0$ into cells.

Consider a cell $\D$.
The number of $(d-2)$-flats intersecting $\D$ is $O(n/r^2)$,
which can be made 0 by setting $r:=\Theta(\sqrt{n})$.
So, inside the cell $\D$,
we see only 1-sided orthants, and their union simplifies to
the complement of a box.
In addition, the number of $(d-1)$-flats intersecting $\D$ is $O(n/r)=O(\sqrt{n})$; in other words, the
breakpoints of the step functions in $H$ relevant to the cell $\D$ is $O(\sqrt{n})$.
We can thus apply Lemma~\ref{lem0} to maximize $H$ over the cell $\D$
in $\OO((\sqrt{n})^2)$ time.
As the number of cells is $O(r^d)=O(n^{d/2})$,
the total running time is $\OO(n^{d/2}\cdot (\sqrt{n})^2)$. 
%(Includes construction time of the partition.)
\end{proof}

\subsection{Algorithm}

We now modify the proof of Theorem~\ref{thm0:anchored} to
solve Problem~\ref{prob:box} for the 2-sided orthant case:

\begin{figure}
\begin{center}
\includegraphics[scale=0.6]{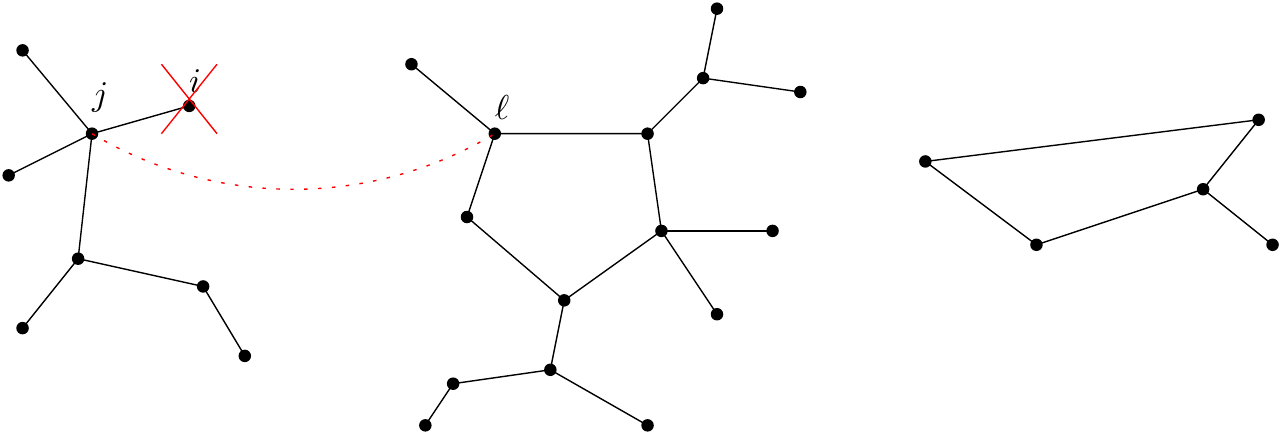}
\end{center}
\caption{After removing vertex $i$ and adding edge
$j\ell$, the graph $G$ remains a pseudo-forest.}\label{fig:graph}
\end{figure}

\begin{theorem}\label{thm0:box}
In the case when all input orthants are 2-sided,
Problem~\ref{prob:box} can be solved in $\OO(n^{d/4+1})$ time for any constant even $d$.
\end{theorem}
\begin{proof}
We maintain a $G$-function $H$.
Initially, $H(x_1,\ldots,x_d)=(\sigma(x_1)+\sigma(x_2))(\sigma(x_3)+\sigma(x_4))\cdots(\sigma(x_{d-1})+\sigma(x_d))$, with $G$ being a matching with $d/2$ edges, where $\sigma(x)$ denotes the successor of $x$ among all $O(n)$ input coordinate values.  We call an index $i$ \emph{free} if $x_i$ appears 
exactly once in $H$
and is ``unaltered'' (i.e., $H$ is of the form $(\sigma(x_i) + h(x_\ell)) \cdots$
where $x_i$ does not appear in ``$\cdots$'').
All indices are initially free.
We maintain the following invariants: at any time,
(i)~$G$ is a pseudo-forest with at most $d/2$ edges,
and (ii)~for each component $T$ of $G$ which is a tree (not a 1-tree), $T$
has at least two free leaves.

In each iteration, we pick a free leaf $i$
in some component $T$ of $G$ which is a tree.
As before, we rewrite the expression $E$ as a disjunction of $O(1)$
subexpressions, where in each subexpression, only two occurrences of $x_i$ remain---in a predicate of the form
$[x_i\le f(x_j)]$, and another predicate of the form $[x_i\ge g(x_k)]$.

We branch off to maximize $H$ for each of these subexpressions separately.
In such a subexpression, to eliminate the variable $x_i$ while maximizing $H$,
we replace the two predicates $[x_i\le f(x_j)]$ and $[x_i\ge g(x_k)]$ with $[f(x_j)\ge g(x_k)]$,
and replace $x_i$ with $f(x_j)$ in $H$ (since $x_i$ is free).  
Now, $i$ and $j$ are not free.
Also, in the graph $G$, the unique edge $i\ell$ incident to $i$ 
is replaced by
$j\ell$ (unless $j=\ell$).  If $j$ is in the same component $T$ as $i$,
then $T$ becomes a 1-tree; otherwise, two components are
merged and the new component is either a tree with at least two free
leaves, or a 1-tree.  (See Figure~\ref{fig:graph}.)
So, the invariants are maintained.

We stop a branch when there are no free indices left.
At the end, we get $O(1)$ subproblems, where in each subproblem,
all components are 1-trees, and so the number of nodes
is exactly equal to the number of edges, implying that
the dimension is $d'\le d/2$.
Now we can apply Lemma~\ref{lem1} to solve these subproblems in $\OO(n^{d'/2+1})$ time.
\end{proof}

\begin{corollary}
Problem~\ref{prob:box}
can be solved in $\OO(n^{(5d+4)/12})$ time for any constant even~$d$.
\end{corollary}
\begin{proof}
Following the proof of Corollary~\ref{thm:anchored} but using
Theorem~\ref{thm0:box} instead of Theorem~\ref{thm0:anchored}
gives running time $\OO(n^{d/3} \cdot (n^{1/3})^{d/4+1})$.
\end{proof}

Applying the above corollary in $2d$ dimensions, we finally obtain:

\begin{corollary}
Given $n$ points in $\R^d$ and a box $B_0$,
we can compute the maximum-volume empty box inside $B_0$
in $\OO(n^{(5d+2)/6})$ time for any constant $d$.
\end{corollary}
%\begin{proof}
%divide-and-conquer to reduce to the origin-containing case...
%\end{proof}

\section{Remarks}\label{sec:rmks}

\paragraph{On the 2D algorithm.} 
The $2^{O(\log^*n)}$ factor can be analyzed more precisely (an upper
bound of $3^{\log^*n}$ can be shown with minor changes to the algorithm).  
A question remains whether the extra factor could be further lowered
to inverse-Ackermann, or eliminated completely.

The previous algorithm by Aggarwal and Suri~\cite{AggSur}
used matrix searching techniques, namely, for finding row minima in certain types of partial Monge matrices.  
We are able to bypass such subroutines 
because we have focused our effort on solving the \emph{decision problem}
(due to the author's randomized optimization technique~\cite{ChaSoCG98}).
Generally, the row minima problem is equivalent 
to the computation of lower envelopes of pseudo-rays and pseudo-segments, not necessarily of constant complexity~\cite{ChaSODA21}.
However, to solve the decision problem, we only need
lower envelopes of pseudo-rays and pseudo-segments of constant complexity
(formed by hyperbolas), for which there are simpler direct methods, as
we have noted in Lemma~\ref{lem:LE}.  (Incidentally, the proof we gave for
reducing Lemma~\ref{lem:LE}(b) to (a) is essentially equivalent to Aggarwal and Klawe's reduction
of row minima in double-staircase to staircase matrices~\cite{AggKla};
a similar idea has also been used in dynamic data structures with ``FIFO updates''~\cite{ChHePr}.)

On the other hand, it should be possible to modify our approach to get improved \emph{deterministic} algorithms for 2D largest empty rectangle, by
solving the optimization problem directly and using known matrix searching subroutines~\cite{KlaKle}, though details are more involved
and the running time seems slightly worse than in our randomized algorithm.

It is theoretically possible to devise an optimal algorithm
for Problem~\ref{prob1} without knowing the true complexity of the algorithm, since by a constant number of rounds of recursion in our
method, the problem is reduced to subproblems of very small size (say, $\log\log\log\log n$),
for which we can afford to explicitly build an optimal decision tree
(this type of trick appeared before in the literature~\cite{Lar,PetRam}).

%\paragraph{On the 3D algorithm.}

%derandomized

\paragraph{On the higher-dimensional algorithms.}
Our approach in higher dimensions works for maximizing the perimeter (sum of edge lengths) of
the box as well.  In fact, the algorithm for the simpler, largest empty \emph{anchored} box problem should suffice here after doubling the dimension, since the required objective function here is $H_{\mbox{\scriptsize\rm perim}}(x_1,\ldots,x_d)=x_1+\cdots+x_d$, which is ``similar'' to $\Hvol(x_1,\ldots,x_d)=x_1\cdots x_d$.

For the largest empty anchored box problem, the $\OO(n^{5d/12})$ time bound can be further improved to
$\OO(n^{(7d+6)/18})$, by building on the graph-theoretic ideas from
Section~\ref{sec:box}, as we show in 
\PAPER{Appendix~\ref{sec:anchored:new}}\LIPICS{the full paper}.  Still further improvements of the exponent is likely possible,
by working with $G$-functions for \emph{hypergraphs} $G$, not just graphs, though improvement on the fraction $7/18$ appears very tiny and requires $d$ to be a very large constant, and the algorithm becomes more complicated.
For the largest empty box problem, we currently don't know how to improve
the fraction $5/6$, even using hypergraphs.
It remains a fascinating question what the best fraction $\beta$
is for which the problem could be solved in $O(n^{\beta d +o(d)})$
time.  

On the conditional lower bound side, another relevant question is whether Problem~\ref{prob:anchored}
or \ref{prob:box}
remain $W[1]$-hard with respect to the parameter $d$ in the special case of 2-sided orthants.

\paragraph{Acknowledgement.}
I thank David Zheng for discussions on the 2D problem.

{
\small
\bibliographystyle{plainurl}
\bibliography{opt_box_arxiv}
}

\PAPER{%%%%%%%%%%%%%%

\appendix

\section{Largest Empty Anchored Box in Higher Dimensions\PAPER{\\} (Further Improved)}\label{sec:anchored:new}

In this section, we return to the largest empty \emph{anchored} box problem and describe a further improvement to the result in Section~\ref{sec:anchored}, by incorporating the graph-theoretic approach
from Section~\ref{sec:box}.

%Improvement for the anchored case:
%$O(n^{(7d+6)/18})$

%for small $d$, e.g., $\OO(n^{7/3})$ for $d=6$

%perimeter

For a multigraph $G$ with vertex set $\{1,\ldots,d\}$ (without self-loops), define 
a \emph{generalized $G$-function} $H:\R^d\rightarrow\R$ to be a function of the form
\[ H(x_1,\ldots, x_d)\ =\  \prod_{i=1}^d h_i(x_i)\cdot \prod_{e=ij\in G} h'_e(x_i,x_j),
\]
where each $h_i$ is a univariate step function and each $h'_e$ is a bivariate step function.
Here, in a bivariate step function $h'_e$, the domain is divided into grid cells
by horizontal and vertical lines, and $h'_e$ is constant in each grid cell; the complexity of $h'_e$ refers to the number of horizontal and
vertical lines.
The \emph{complexity} of $H$ is the total
complexity of the univariate and bivariate step functions.

\begin{lemma}\label{lem0:new}
Let $H$ be a generalized $G$-function with $O(n)$ complexity, where $G$ is a matching.  Given a box $B_0$, we can compute the maximum of $H$ over $B_0$ in $O(n^2)$ time for any constant $d$.
\end{lemma}
\begin{proof}
Trivial.
\end{proof}

\begin{lemma}\label{lem1:new}
Let $H$ be a generalized $G$-function with $O(n)$ complexity, where $G$ is a matching.
Given a set $S$ of $n$ boxes in $\R^d$ and a box $B_0$, 
we can compute the maximum
of $H$ in $B_0-\bigcup S$ in $O(n^{d/2+1})$ time for any constant $d$.
\end{lemma}
\begin{proof}
Similar to the proof of Lemma~\ref{lem1}, but using Lemma~\ref{lem0:new}
instead of Lemma~\ref{lem0} as subroutine.
\end{proof}

We now improve Lemma~\ref{lem1:anchored} for 2-sided orthants:

\begin{lemma}\label{lem1:anchored:new}
Let $H$ be a simple function with $O(n)$ complexity.
Given a set $S$ of $n$ 2-sided orthants in $\R^d$ and a box $B_0$,
we can compute the maximum
of $H$ in $B_0-\bigcup S$ in $\OO(n^{d/3+1})$ time  for any constant $d$.
\end{lemma}
\begin{proof}
We maintain a generalized $G$-function $H$.
Initially, $G$ consists of $d$ isolated vertices.  We repeatedly find variables $x_i$ to eliminate:

\begin{itemize}
\item 
{\sc Case 1:}  there is an isolated index $i$ in $G$.
As before, we rewrite the expression $E$ as a disjunction of $O(1)$
subexpressions, where in each subexpression, only two occurrences of $x_i$ remain---in a predicate of the form
$[x_i\le f(x_j)]$, and another predicate of the form $[x_i\ge g(x_k)]$.

We branch off to maximize $H$ for each of these subexpressions separately.
In such a subexpression, to eliminate the variable $x_i$ while maximizing $H$,
we replace the two predicates $[x_i\le f(x_j)]$ and $[x_i\ge g(x_k)]$ with $[f(x_j)\ge g(x_k)]$,
and replace $h_i(x_i)$ with $h'(x_j,x_k) := \min_{x: g(x_k)\le x\le f(x_j)} h_i(x)$ in $H$.  
We remove $i$ from $G$, and add edge $jk$ to $G$.
% (Could have parallel edges.)
%Note that $j$ and $k$ are no longer isolated (if they were before).

\item
{\sc Case 2:} there is an index $i$ of degree at least 2 in $G$.
We try out all $O(n)$ different settings of $x_i$ (breakpoints
of the step functions), and obtain $O(n)$ instances in which
$i$ is removed from~$G$.
\end{itemize}

We stop a branch when neither case is applicable, i.e., all indices in $G$ have
degree 1, i.e., $G$ is a matching.  Here, we can apply Lemma~\ref{lem1:new} to solve the problem.

Consider one branch.
Suppose Case 1 is applied $s$ times and Case 2 is applied $t$ times.
At the end, the number of vertices in $G$ is $d':=d-s-t$, and
the number of edges in $G$ is at most $s-2t$ (since Case 1 adds one
edge and Case 2 removes at least two edges).
Since $G$ is a matching at the end, the number of vertices is twice the number of edges.  Thus,
$d-s-t\le 2(s-2t)$, i.e., $s\ge d/3+t$, i.e., $d'\le 2d/3-2t$.
Since $O(n^t)$ instances are generated and Lemma~\ref{lem1:new}
has cost $O(n^{d'/2+1})$,
the total cost is $\OO(n^t\cdot n^{d'/2+1})
\le \OO(n^{d/3+1})$.
\end{proof}

The above lemma improves the time bound in Theorem~\ref{thm0:anchored}
to $O(n^{\down{d/2}/3+1})$.
This in turn improves the time bound in Corollary~\ref{thm:anchored}
to $O(n^{d/3}\cdot (n^{1/3})^{\down{d/2}/3+1}) 
= O(n^{(d+1)/3 + \down{d/2}/9})$.

\begin{corollary}
Given $n$ points in $\R^d$ and a box $B_0$,
we can compute the maximum-volume empty anchored box inside $B_0$ in $\OO(n^{(d+1)/3 + \down{d/2}/9})\le \OO(n^{(7d+6)/18})$ time for any constant~$d\ge 3$.
\end{corollary}

}%%%%%%%%%%%%%%%%%%%%%%%%%%%%%%%%%%%%

\end{document}